\documentclass[12pt,reqno]{article}
\usepackage{amsthm}
\renewcommand{\a}{\alpha}

\newcommand{\g}{\gamma}

\renewcommand{\d}{\delta}
\renewcommand{\O}{\Omega}
\def\cd{{\mathcal D}}

\def\ce{{\mathcal E}}

\newcommand\eps{\varepsilon}
\renewcommand{\epsilon}{\varepsilon}
\newcommand{\DD}{{\cal D}}

\newcommand{\etaeps}{\eta_\epsilon}
\newcommand{\etabar}{\bar{\eta}}%
\newcommand{\etatil}{\tilde{\eta}}%

\newcommand{\laeps}{{\lambda_\eps}}%

\newcommand{\gradp}{\nabla^\perp}
\newcommand{\Om}{\Omega}

\newcommand{\dist}{{\rm dist\, }}

\newcommand{\la}{\lambda}

\def\R{{\bf R}}

\def\ep{\varepsilon}
\def\lep{{|\mathrm{log }\ \ep|}}

\def\dist{{\rm dist}}

\def\O{\Omega}

\def\a0{\alpha_0}
\def\a{\alpha}

\def\g{\gamma}

\def\r0{\rho_{0}}

\def\la{\lambda}

\def\lep{{|\log \ep|}}

\newtheorem{thm}{Theorem}[section]
\newtheorem{lem}[thm]{Lemma}

\newtheorem{rem}[thm]{Remark}

\newcommand{\be}{\begin{equation}}
\newcommand{\ee}{\end{equation}}
\newcommand{\bea}{\begin{eqnarray}}
\newcommand{\eea}{\end{eqnarray}}
\newcommand{\beann}{\begin{eqnarray*}}
\newcommand{\eeann}{\end{eqnarray*}}
\newcommand{\nnn}{\nonumber}

\begin{document}
\title{Non existence of vortices in the small density region of a condensate}
\author{{\Large Amandine Aftalion}\footnote{CMAP, CNRS, Ecole Polytechnique, 91128 Palaiseau cedex, France}
 \and {\Large Robert L.Jerrard}\footnote{Dept. of Mathematics  University of Toronto, Toronto, Canada
M5S2E4.}
\and {\Large Jimena Royo-Letelier}\footnote{CMAP, Ecole Polytechnique, 91128 Palaiseau cedex, France.}}   \thispagestyle{empty}\maketitle

\begin{abstract}

In this paper, we answer a question raised by Len Pitaevskii and
prove that the ground state of the Gross Pitaevskii energy
describing a Bose Einstein condensate at low rotation does not have
vortices in the low density region. Therefore, the first ground
state with vortices has its vortices in the bulk. This is obtained
by proving that for small rotational velocities, the ground state is
multiple of the ground state with zero rotation. We rely on
sharp bounds of the decay of the wave function combined with
weighted jacobian estimates.
\end{abstract}

\section{Introduction}

Among the many experiments on Bose Einstein condensates, one
consists in rotating the trap holding the atoms in order to observe
a superfluid behaviour: the appearance of quantized vortices
\cite{AK,SB,PS,PiSi,A}. This takes place for sufficiently large
rotational velocities. On the contrary, at low rotation, no vortex
is detected in the bulk of the condensate. The system can be
described by a complex valued wave function minimizing a Gross
Pitaevskii type energy. A vortex corresponds to zeroes of the wave
function with phase around it. The density of the condensate is
significant in a region which is either a disk or an annulus, and
gets exponentially small outside this domain. Vortices are
experimentally visible in the bulk of the condensate. A question
raised by Len Pitaevskii is whether for small rotational velocity,
when there are no vortices in the bulk, vortices could exist in the
low density region.  For very large rotational velocities, when bulk
vortices are arranged on a triangular lattice, it has been shown
\cite{ABN2} that the vortex distribution extends to infinity.
Therefore, in this case, there are many vortices in the low density
region. It is then very natural to wonder whether vortices first
appear in the bulk or at infinity. It is experimentally and
numerically difficult to observe a vortex, which is a zero, in a low
density region. Mathematically this could not be achieved through
energy estimates or expansion since the contribution of a vortex in
a low density region is very small. In this paper, we introduce new
ideas to answer Pitaevskii's question and prove that at low
velocity, there are indeed no vortices in the condensate, even in
the low density region. Therefore, the first ground state with
vortices has its vortices in the bulk.

Since a condensate is a trapped object, the geometry of the trap
plays a role. An important special case is a radial harmonic
trapping potential $V(r)=r^2$. The space can then be split into two
regions, a region of the form $\cd=\{\lambda_0>V(r)\}$ (for a
suitable constant $\lambda_0$), where the wave function is
significant and the condensate is mainly located, and a region
$\R^2\setminus \cd$ where the modulus of the wave function is
exponentially small \cite{A}. In this latter region, it is very
difficult to determine mathematically the contribution of a vortex
to the energy. Ignat and Millot \cite{IM1,IM2} have determined the
critical rotational velocity $\Omega_c$ for the nucleation of the
first vortex inside $\cd$. This theorem does not describe the
behaviour in $\R^2\setminus \cd$. A natural question is whether for
$\Omega <\Omega_c$, the minimizer of the energy has zeroes in this
region, whether there is a smaller critical velocity than $\Omega_c$
where the minimizer is unique and vortex free. At very high
velocity, it has been proved in \cite{ABN2} that vortices exist up
to infinity so it seems reasonable that  at smaller velocity,
vortices may exist in the exponentially small region, far away from
the bulk and could arrange themselves on disks or arrays close to
infinity. In fact, we prove that this is not the case before
$\Omega_c$, namely that the minimizer is unique and does not vanish.
It means that for a large range of rotational velocities $\Omega$,
the minimizer is given by the same function.

 We consider here a two-dimensional setting and define the energy
for the complex-valued wave function $u$, such that $\int _{\R^2}
|u|^2=1$, as
\be\label{Eep}
E_\ep(u)= \int_{\R^2} \left\{  \frac12
|\nabla u|^2
  + {1\over 4\eps^2}|u|^4+ {1\over 2\eps^2}V(x)|u|^2
 - \Om x^\perp \cdot (iu,\nabla u) \right\} dx,
 \ee
 where $\Om$ is the angular velocity, $x=(x_1,x_2)$,
$x^\perp=(-x_2,x_1)$, $\eps>0$ is a small parameter, $V(x)$ is the
  trapping potential and $(iu,\nabla u)=iu\nabla u^*-iu^*\nabla u$.
The class of potentials includes the model case $V = x_1^2 + x_2^2$.
Then, the critical angular velocity for nucleation of vortices is of
order $\lep$ (see \cite{IM1}). An upper bound on the rotational
velocity is given by $\O<1/\eps$ when the confinement breaks down.
  The condensate is mostly concentrated in the
  region \be \cd := \{ x\in \R^2 : V <\lambda_0\} \label{cd.def}\ee
 where $\lambda_0$ is chosen so that
  \be \int_{\R^2} (\lambda_0-V(x))^+ dx = 1. \label{l0.def}\ee   We refer to \cite{A}
for more details on how this is derived from the physical
experiments.

In recent experiments in which a laser beam is superimposed upon the
magnetic trap holding the atoms, the trapping
 potential $V(x)$ is of a  different  type
\cite{R,SB,W,AM}:\be\label{pot-gaus} V (r)= r^2+V_0e^{-r^2/w_0}.
   \ee When the gaussian is expanded around the origin,
    this leads to a harmonic plus quartic potential \cite{SB}\be\label{pot-approx} V (r)=(1-b) r^2+
{k\over 4} r^4.
   \ee If $b$ is small ($b<1+(3k^2/4)^{1/3}$),
the domain $\cd$ given by (\ref{cd.def}) is a disc, while if
$b>1+(3k^2/4)^{1/3}$, it is an annulus.   According to the values of
$V_0$ and $w_0$ in the case of (\ref{pot-gaus}), the domain $\cd$
can also be a disk or an annulus.

In this paper, we consider potentials $V$ including $r^2$ and of the
type (\ref{pot-gaus}) or (\ref{pot-approx}) when the bulk $\cd$ is a
disk. In the case where $\cd$ is a disk, the potential $V$ is not
necessarily required to be increasing.

\subsection{Assumptions}\label{S:assumptions}

Throughout this paper, we make the following assumptions about the
potential $V$. First, \be V\mbox{ is nonnegative and radial, }\ V\in
C^1, \label{V1}\ee and \be \mbox{ there exists $c_0> 0, p\ge 2$ such
that } \frac 1{c_0} r^p \le V(r) \le c_0 r^p \quad\quad\mbox{ if }r
\ge c_0. \label{Vgrowth}\ee This assumption is easily seen to imply
that $E_\ep$ is bounded below for $|\Omega| \ll \frac 1 \ep$ and
that the angular momentum term $x^\perp \cdot (iu, \nabla u)$ 
is integrable as long as $u$
has finite energy. We will also use (\ref{Vgrowth}) to obtain decay
estimates that justify for example the integration by parts leading
to a decoupling of the energy. We fix $\lambda_0\in \R$ such that
(\ref{cd.def})-(\ref{l0.def}) hold. Such a $\lambda_0$ exists due to
the growth of $V$.  We further assume that the bulk $\cd$ is a disk
and not an annulus, that is $V$ is such that \be \cd = B_R(0)\mbox{
for some }R>0 \label{simplyconnected}\ee and that there exist
$\delta_0>0$ and a $C^1$ function $R: (-2\delta_0,2\delta_0)\to \R$
also denoted $R_\delta=R(\delta)$, such that \be
\begin{array}{l}
R_0 = R \,, \qquad \{ x : V(x)<\lambda_0+\delta \} = B_{R_{\d}}(0) \\ \\
\mbox{and} \quad  0< \frac 1C \le dR/d\lambda \le C \quad  \mbox{for
some constant   } C
\end{array}
\label{V2} \ee where $B_r(y)$ denotes the open ball of radius $r$
about $y$. This implies that $\lambda_0 - V$ is bounded away from
$0$ in the interior of $\cd$; in physical terms, this assumption
rules out the case of annular bulks and ``giant vortices" at low
angular velocities. We remark that the assumption above implies that
if $|x| \in (R_{-\d},R_{\d})$ and $0\le \d \le \d_0$ then
$\dist(x,\partial \cd) = \mathcal{O}(\delta)$.

We point out that assumptions (\ref{Vgrowth})  and (\ref{V2}) imply
that \be \mbox{ there exists $c_1>0$ such that $V(r)- \lambda_0 \ge
c_1(r^2 - R^2)$ for all  $r\ge R$.} \label{Vquad}\ee

Our assumptions include indeed potentials like $r^2$ or
(\ref{pot-approx}) for a disk case, and do not require $V$ to be
increasing.

\subsection{Main result}

Our main result is
\begin{thm}
Assume that $u_\ep$ minimizes $E_\ep(\cdot)$ with rotation $\Omega$,
and let $\etaeps$ denote the minimizer of $E_\ep(\cdot)$ for $\Omega
= 0$.  There exists $\ep_0,\omega_0, \omega_1>0$ such that if $0<\ep<\ep_0$
and $\Omega\le \omega_0 |\log\ep| - \omega_1 \log|\log\ep|$ then $u_\ep
= e^{i\alpha} \etaeps$  in $\R^2$ for some constant $\alpha$.
\label{Tfull}\end{thm}

In the pure quadratic case  $V= r^2$, Ignat and Millot \cite{IM1,
IM2} have shown that the {\em bulk} of the condensate (that is any
domain contained in $\cd$) is vortex-free for $|\Omega| \le
\omega_0|\log \ep| -\omega_1 \log|\log\ep|$, for some
$\omega_1>0$ and the same constant $\omega_0$ that we find in Theorem
\ref{Tfull} . They have no information on what happens in
$\R^2\setminus \cd$. Our theorem proves that vortices do not lie in
$\R^2\setminus \cd$. They have also shown that there exists
$\delta>0$ such that the ground state has at least one vortex in the
bulk if $\Omega \ge \omega_0|\log \ep| + \delta \log|\log\ep|$. In
this sense, our estimate $|\Omega|\le \omega_0|\log \ep| -\omega_1
\log|\log\ep|$ captures the sharp leading-order term, and the
correct scaling of the next-order term, of the critical velocity for
vortex formation. We point out that our arguments also deal with
more general potentials. The arguments used in \cite{IM1} to prove
the existence of interior vortices for rotations greater than $
\omega_0|\log \ep| + \delta \log|\log\ep|$ should extend with few
changes to the more general potentials considered here, using
results about auxiliary functions that we establish in Section 3 in
place of parallel results from \cite{IM1}. Thus the constant
$\omega_0$ should also be sharp for these more general potentials.

We split the proof into two independent results. The first main
result of this paper asserts roughly speaking that {\em symmetry
breaking occurs first in the interior of $\cd$}: if $\Omega$ is
small enough that there are no vortices in $\cd$, then there are no
vortices anywhere.

\begin{thm}
Assume that $u_\ep$ minimizes $E_\ep(\cdot)$ with rotation $\Omega$, and let $\etaeps$ denote
the minimizer of $E_\ep(\cdot)$ for $\Omega = 0$. Assume also that $\Omega\le C |\log\ep|$ for some $C$.

There exists $\ep_0>0$ such that if $0<\ep<\ep_0$ and  $\Omega$ is subcritical in the sense that
\be
|u_\ep| \ge \frac 12 \eta_\ep\mbox{ in }\cd_1 := \{x\in \cd: \mbox{dist}(x,\partial \cd) \ge |\log\ep|^{-3/2}\}
\label{subcrit}\ee
then $u_\ep = e^{i\alpha} \etaeps$  in $\R^2$ for some constant $\alpha$.
\label{T1}\end{thm}

Our second main  theorem gives an estimate for the critical value of $\Omega$. 
The statement of the theorem refers to an auxiliary function $f_0$:
let $$a(x)=\lambda_0-V(x),$$
\begin{equation}
\eta_0 := \sqrt {a^+}, \quad\quad \xi_0(r) = \int_r^\infty s \eta_0^2(s) \ ds, 
\quad\quad f_0(r) = \left\{\begin{array}{ll}0&\mbox{if }r\ge R\\
\xi_0(r)/ \eta_0^2(r) &\mbox{if }r \le R;\end{array}\right.
\label{f0.def}\end{equation}

\begin{thm}
Let $\omega_0 = \frac 1{2 \|f_0\|_\infty}$. There exists
$\omega_1>0$ and $\ep_1>0$ such that if $|\Omega| \le  \omega_0|\log
\ep| -\omega_1 \log|\log\ep|$ and $0< \ep<\ep_1$, then $\Omega$ is
subcritical in the sense of (\ref{subcrit}), and  the conclusion of
Theorem \ref{T1} thus holds. \label{T2}\end{thm}

In our proof of Theorem \ref{T2}, as in estimates of the critical
rotation in works such as  \cite{IM1} and \cite{AAB}, a main point
is to obtain sharp energy lower bounds. In all earlier works that we
know of, this is done using the vortex ball construction originally
introduced by \cite{je} and \cite{sa}. In our proof of Theorem
\ref{T2}, we avoid any explicit\footnote{However, the proof of Lemma
\ref{L.newJacest}, see Lemma 8 in \cite{J}, ultimately relies on a vortex ball
construction appearing in \cite{js}.} mention of vortex balls by
instead appealing to a result from \cite{J}, stated here as Lemma
\ref{L.newJacest}. This makes our argument considerably shorter than
those in \cite{AAB, IM1} and other references.

We point out that the results of \cite{IM1,  IM2}  do {\em not}
directly imply that Theorem \ref{T2} holds in the case $V = r^2$,
although it is possible that this conclusion can be extracted with
relatively little effort from arguments in these references.

\subsection{Main ideas of the proof}

 The energy minimizers with
$\O=0$ provide real solutions to the Euler-Lagrange equations: when
$\O=0$, $E_\ep(\eta)=G_\ep(\eta)$, where
\be\label{G}
G_\ep(\eta)= \int_{\R^2} \left\{  \frac12 |\nabla \eta|^2
 +{1\over 4\ep^2}|\eta|^4+ {1\over 2\ep^2}V(x)|\eta|^2
\right\}
  dx.
\ee Our main goal consists in proving that up to the critical
velocity of nucleation of bulk vortices, the minimizer of $E_\ep$
with velocity $\Omega$ is in fact equal to $\etaeps$.

The minimizer $\etaeps$ of $G_\ep$ under the $L^2$ constraint of
norm 1, is (up to a complex multiplier of modulus one) the unique
positive solution of \be\label{1-1} -\Delta \etaeps +
{1\over\ep^2}\etaeps(V(x)+\etaeps^2) = {1\over \ep^2} \lambda_\ep
\etaeps \ee where ${1\over \ep^2} \lambda_\ep$ is the Lagrange
multiplier, which is also necessarily unique. Moreover, $\lambda_\ep
\to \lambda_0$, and  $\etaeps ^2$ converges to $a^+$ in $L^2(\DD)$
and uniformly on any compact set of $\cd$. We will need some
estimates on the decay of $\etaeps$ at infinity that we prove in
section 2.

By a remarkable identity (see  Lassoued \& Mironescu \cite{LM}), for
any $u$, the energy $E_\ep$ for any $\Omega$ splits into two parts,
the energy $G_\ep(\etaeps)$ of the density profile and a reduced
energy of the complex phase $v=u/\etaeps$:
 \be\label{split}
  E_\ep(u)=G_\ep(\etaeps)+F_\eps(v)
\ee \be\label{Geps} \hbox{where}\quad
F_\eps(v)= \int_{\R^2} \left\{
{\etaeps^2\over 2} |\nabla v|^2
+ {\etaeps^4\over 4\eps^2} (|v|^2-1)^2-
 \etaeps^2 \Om { x}^\perp\cdot (iv,\nabla v) \right\}\ dx.
\ee In particular the potential $V(x)$ only appears in $G_\ep$. We will recall the
proof of (\ref{split}), as well as that of (\ref{Fep}) below, in Section \ref{S:split}.
This kind of splitting of the energy is by now standard in the rigorous analyses
of functionals such as $E_\ep$.

Next, define \be\label{eqxi}
 \xi_\ep (r)= \int_r^\infty s\etaeps^2 (s)\ ds,
\ee so that  $\gradp \xi_\ep = x^\perp \etaeps^2$. An integration by
parts 
yields
\be\label{Fep}F_\eps(v)= \int_{\R^2} \left\{ {\etaeps^2\over 2}
\left (|\nabla v|^2-\frac{4\O\xi_\ep}{\etaeps^2} Jv \right ) +
{\etaeps^4\over 4\eps^2} (|v|^2-1)^2 \right\}\ dx\ee where $J v =
\frac 12 \nabla\times (iv, \nabla v) = (i v_{x_1}, v_{x_2})$ is the
Jacobian.

We recall that the function $f_\ep := \xi_\ep/\etaeps^2$ appearing
in $F_\ep$ is important since it is well known that vortices in the
interior of $\cd$ first appear near where this function attains a
local maximum \cite{A,AAB,IM1,IM2}; its importance is also clear from
(\ref{Fep}), since it controls the relative strength of the positive and
negative contributions to $F_\ep$. The proofs of Theorems \ref{T1}
and \ref{T2} rest on new bounds for $f_\ep$ in $\R^2\setminus \cd$
and near $\partial \cd$, which in
turn rely on decay estimates for $\etaeps$. In particular, we show in Lemma \ref{lem:f}
that $f_\ep \le C\ep^{2/3}$ in $\R^2\setminus \cd$.


The other part of the proof consists essentially of bounds of
$2\Omega \int \etaeps^2 f_\ep Jv$ by the positive terms in $F_\ep$.
Away from the bulk,  we use our estimates of $f_\ep$
to find that $2 \Omega f_\ep Jv$ is bounded pointwise by $\frac 12 |\nabla v|^2$.
In the bulk, where $\etaeps^2$ is not too small,  we have
\[
\frac 12 \etaeps^2|\nabla v|^2 + \frac {\etaeps^4}{4\ep^2}(|v|^2-1)^2 \ge
 \etaeps^2[ \frac 12|\nabla v|^2 + \frac 1 {4\tilde \ep^2}(|v|^2-1)^2]
\]
for some $\tilde\ep$ such that $|\log \tilde \ep| = | \log\ep| (1+o(1))$. We obtain
the desired bounds by combining this
with a weighted Jacobian estimate mentioned above, Lemma  \ref{L.newJacest},
which {\em directly} implies that
\[
2\Omega \int  \chi \etaeps^2 f_\ep Jv
\le
 \Omega\left(\frac {2\|f_\ep\|_{\infty}}{\log \tilde \ep|}\right)\int \chi \etaeps^2  [ \frac 12|\nabla v|^2 + \frac 1 {4\tilde \ep^2}(|v|^2-1)^2] +
\mbox{ small error terms}
\]
where $\chi$ is a cutoff function supported in the bulk. Note that the leading-order critical rotation
$\omega_0$ is such that $\Omega\left(\frac {2\|f_\ep\|_{\infty}}{\log \tilde \ep|}\right)\approx \Omega/ \omega_0|\log\ep|$. The proof of Theorem \ref{T2} is completed by assembling these ingredients and
controlling error terms. The proof of Theorem \ref{T1} relies on an additional ingredient, which is
that if $|v|\ge \frac 12$ in an open set $U$, then $Jv$ is extremely close in $U$ to $J(\frac v{|v|})=0$.
Theorem \ref{Tfull} follows immediately from combining Theorems \ref{T1} and \ref{T2}.

An interesting open problem is to see to what extent this analysis continues to hold if the assumption of radial symmetry is dropped. In our arguments, this symmetry is used heavily in our analysis of the behavior of $f_\ep$ away from the bulk, and near the boundary of the bulk.

\hfill


We briefly remark on the assumption (\ref{Vgrowth}) of quadratic growth. Our proofs
show that the absence of vortices in the low density region is a
consequence of the fact that the auxiliary function $f_\ep =
\xi_\ep/\etaeps^2$ is very small in $\R^2\setminus \cd$. The proof
of this fact (see see Lemma \ref{lem:f}) can be modified to show
that if for example (\ref{Vgrowth}) holds with $p<2$, then $f_\ep(r)
\ge C \ep r^{1 - p/2} \to \infty$ as $r\to \infty$. However, in this
situation $E_\ep$ is unbounded below for any $\Omega\ne 0$.  This
reflects the fact that a subquadratic trapping potential is not
strong enough to contain a rotated condensate.



\section{Properties of auxiliary functions}

In this section we study the real-valued minimizer $\etaeps$ and the auxiliary functions
 $\xi_\ep$ and $f_\ep = \xi_\ep/\etaeps^2$ defined as
 \begin{equation}
 \xi_\ep(r) = \int_r^\infty s \eta_\ep^2(s) \ ds, 
\quad\quad f_\ep(r) = \xi_\ep(r)/ \eta_\ep^2(r).
\label{fep.def}\end{equation}

\begin{thm} Assume that $V$ satisfies (\ref{V1}), (\ref{V2}). Then for every $\ep>0$, there exists a unique positive minimizer $\etaeps$ of $G_\ep$ in
\[
\mathcal{H} := \{ u\in H^1(\R^2) :  \ \int_{\R^2}  |u|^2V(x) <\infty,\ \  \int_{\R^2} |u|^2 = 1 \}.
\]
Every minimizer of $G_\ep$ in $\mathcal{H}$  has the form
$e^{i\alpha}\etaeps$, for $\alpha$  constant.
Moreover, $\etaeps$ is a radial smooth positive function and  satisfies (\ref{1-1})
with \be |\lambda_\ep -\lambda_0 |\le  C \ep |\log \ep|^{1/2}
\label{etaeps4}\ee where $\lambda_0$ is defined by (\ref{l0.def}).
Finally, recall the notations $R_\delta$ from (\ref{V2}) and
$a=\lambda_0-V$,
 the following estimates are satisfied:
\bea
\eta_\ep(r) &\le & C \, \ep^{1/6} \,  e^{\, c\, \ep^{-1/3}(\sqrt{R}-\sqrt{r})}
\hspace{1.7cm} \mbox{ in }\R^2\setminus \cd
\label{etaeps1}\\
| \etaeps - \sqrt {a^+} | &\le &  C \ep^{1/3} \sqrt {a^+} \hspace{3.5cm} \mbox{ in }B_{R_{-\ep^{1/3}}} \label{etaeps2}\\
\| \nabla \etaeps\|_{L^\infty(\R^2)}
&\le&
C\ep^{-1}.
\label{etaeps3}\\
\etaeps'(r) &\le &0 \quad  \qquad\qquad \qquad \mbox{ for all }r\in (R_{-\d_0}, R_{\d_0})\label{etaeps6}
\\
|\etaeps'(r)| &\le &\frac C \ep  \etaeps(r) \sqrt{ V(r)}\quad\mbox{ for all sufficiently large }r
\label{etaeps7}
\eea
if $\ep < \ep_0$.

\label{theoeta}\end{thm}

Certain parts of the proof follow quite closely arguments given in
\cite{AAB} and in the pure quadratic case in \cite{IM1}.
 Note that some arguments in \cite{IM1} rely strongly on the special
 shape of the potential and cannot be generalized to other
 functions.
 Since $V$ is not necessarily increasing, we have property
 (\ref{etaeps6}) only in the neighborhood of $\partial \cd$.

\begin{proof}
{\em Step 1. Existence of minimizers}:
This follows from standard arguments once we notice that $\int_{\R^2}|u_n|^2 V dx \le C$
is uniformly bounded for any sequence $(u_n)$ minimizing $G_\ep$, and the set
of functions in $\mathcal{H}$ satisfying such a uniform bound is precompact with respect to
weak convergence in $H^1(\R^2)$. This last fact is proved by straightforward and well-known arguments, such as are explained in the proof  in   \cite{IM1}, Lemma 2.1, for $V$ quadratic, the point being that the bound on $\int |u|^2 V$ prevents mass escaping to $\infty$. Standard theory then implies that any minimizer is smooth.
If $\eta$ is any minimizer, then $|\eta|$ is as well, since $G(|\zeta|)\le G(\zeta)$ for all $\zeta$. The strong maximum principle then  implies that $|\eta|$ (and hence $\eta$) never vanishes, and since $G(\eta)\le G(|\eta|)$, it is easy to see that $\eta/|\eta| = e^{i\alpha}$ for some constant $\alpha$. We henceforth let $\etaeps$ denote a  fixed positive minimizer.\\

{\em Step 2: uniqueness of $\etaeps$}. Multiplying (\ref{1-1}) by
$\etaeps$ and integrating by parts  we find that $\mu_\ep $ is
positive. Suppose that there are two couples $(\eta_0,\mu_0)$ and
$(\eta_1,\mu_1)$ satisfying (\ref{1-1}) such that $\|
\eta_0\|_{L^2}=1=\|\eta_1\|_{L^2}$ and $\mu_0>\mu_1$, and define
$w=\frac{\eta_1}{\eta_0}$. This function verify

\[
\int_{\R^2} \eta_0^2 (w-1)^2 \, dx = 2 \int_{\R^2} (\eta_1^2 - \eta_0\eta_1 ) \,dx =  2 \int_{\R^2}  \eta_0^2 w (w-1) \, dx
 \]
and
\[
-\nabla \cdot (\eta_0^2 \nabla w ) + \frac 1{\ep^2}  \eta_0^4 w (w^2-1) = (\mu_1-\mu_0) \eta_0^2 w .
 \]

Multiplying the second equality by $(w-1)$, integrating by parts and then using the first equality we find
\[
\int_{\R^2} \left\{  \eta_0^2 |\nabla (w-1)|^2 +  \frac 1{\ep^2}  \eta_0^4 w (w-1)^2(w+1) + \frac12 (\mu_0-\mu_1) \eta_0^2 (w-1)^2 \right\}\, dx = 0 .
 \]
The integration by parts is justified in view of (\ref{etaeps1}), (\ref{etaeps7}), which apply to both $\eta_0$ and $\eta_1$, and the proofs of which do not rely on the uniqueness of the minimizer. Hence $w\equiv1$ and $\mu_0=\mu_1$.\\

{\em Step 3: estimate of $\lambda_\ep - \lambda_0$}.
We next note, following standard arguments, that
$G_\ep$ can be rewritten
\[
G_\ep(\eta) =
\int_{\R^2} \left\{  \frac12 |\nabla \eta|^2
+{1\over 4\eps^2}(\eta^2 - a^+)^2 + \frac 1{2\ep^2} a^- \eta^2
        \right\}
  dx +  \frac 1{2\ep^2} \left(\lambda_0 -  \frac 12\int (a^+)^2\right)
\]
if $\|\eta\|_2=1$.  Let $G_\ep^1(\eta)$ denote the first integral above.  We claim that
\[
G_\ep^1(\etaeps) \le C|\log\ep|.
\]
Since $\etaeps$ is a minimizer, to prove this it suffices to construct a competitor for which $G^1_\ep$ is suitably small. To do this, define
\[
g_\ep(s) := \left\{\begin{array}{ll}
\frac s \ep&\mbox{ if }s\le \ep^2\\
\sqrt{s}&\mbox{ if }s\ge \ep^2,
\end{array}\right.
\quad\quad\mbox{ and } \tilde \etaeps := \frac{ g_\ep(a^+)}{\| g_\ep(a^+)\|_{L^2}}.
\]
Note that
\[
1 = \int a^+ \ge \ \int g_\ep^2(a^+) \ = \ \int a^+ - \int_{a^+\le\ep^2} a^+ \{ 1-\frac{a^+}{\ep^2} \}  \ge 1 - C\ep^2.
\]
Using this and explicit calculations such as those in \cite{J},
Lemma 12, the claim is easily verified. We now multiply (\ref{1-1})
by $\etaeps$, integrate by parts and rewrite, recalling the $L^2$
constraint, to find that
\begin{eqnarray}
\frac 1 {\ep^2} (\lambda_\ep - \lambda_0)
&=&
\int |\nabla \etaeps|^2 + \frac 1{\ep^2}(\etaeps^2 -(V-\lambda_0))\etaeps^2 \ dx
\label{mu_ep}\\
&=&
\int |\nabla \etaeps|^2 + \frac 1{\ep^2}(\etaeps^2 - a^+ +a^-)\etaeps^2 \ dx\nnn\\
&=&
\int |\nabla \etaeps|^2 + \frac 1{\ep^2}\left[ a^-\etaeps^2  +(\etaeps^2 - a^+)^2  +  (\etaeps^2 - a^+)a^+\right]dx\\
&\le&
4G_\ep^1(\eta_\ep) +  \frac 1{\ep^2}\| \etaeps^2-a^+\|_{L^2}\|a^+\|_{L^2} \le
C[G_\ep^1(\eta_\ep)  + \frac 1 \ep\sqrt{G_\ep^1(\eta_\ep)}].
\nnn\end{eqnarray}
Thus we have proved (\ref{etaeps4}).\\

{\em Step 4 :  estimates of $\etaeps$}.\\

We claim that

\be
\etaeps^2 \leq \max_{\cd} (\laeps-V) =: A
\label{etainf}
\ee
%

To see this define $w=\frac1{\ep} (\etaeps-\sqrt{A})$. We have that $\etaeps \in L^3_{loc}$, so after (\ref{1-1}) $w,\Delta w \in L^1_{loc}$. Kato's inequality gives $\Delta w^+ \geq sgn^+(w) \Delta w $. Using (\ref{1-1}) again we find

\begin{eqnarray*}
 \Delta w^+ \geq \frac{sgn^+(w)}{\ep^3}\etaeps (\etaeps^2-A)  =   \frac{sgn^+(w)}{\ep^3} (\ep w+\sqrt{A}) (\ep^2 w^2 +2 \ep w \sqrt{A}) \geq (w^+)^3 \quad \mbox{ in } \cd'
\end{eqnarray*}

Hence we have $ -\Delta w^+ +(w^+)^3 \le  0$ in $\cd'(\R^2)$ and $w \in L^3_{loc}$, so using Lemma 2 in \cite{Br}, $w^+\equiv0$. \\

We remark that the properties of the potential $V$ at the boundary (\ref{V2})  implies that the maximum of $\lambda_\ep - V$ is attained at an interior point $x_0$ of $\cd$  such that $\dist(x_0,\partial \cd) > c\, \d_0$.\\

The minimizer being a solution of (\ref{1-1}) in $L^{\infty}$, by elliptic regularity  we derive   that it is a smooth function. \\

{\em Proof of (\ref{etaeps1}).} We construct a supersolution of (\ref{1-1})  of the form \\

\[
\etabar (x ):= \left\{\begin{array}{ll}
\sqrt{\la_0-V(x)+8\d} &\mbox{ if } |x| \leq R_{-\d } \\
\\
\frac{\la_0-\d-V(x)}{6\sqrt{\d}} + 3\sqrt{\d} &\mbox{ if }  R_{-\d }\leq |x| \leq R_{\d } \\
\\
\g \, e^{-\frac{\sqrt{ |x| }}{\sigma}} &\mbox{ if }  R_{\d }  \leq  |x|
\end{array}\right.
\]
where $0<\d<\d_0$ is small parameter that will be determined later and $\g$, $\sigma$ are chosen such that $\etabar \in C^1(\R^2)$, i.e.,

\[
\g = \frac{8 \sqrt{\d}}3  \, \exp\left( \sigma^{-1} R_{\d}^{1/2} \right) \quad  \mbox{ and }\quad\sigma = 16 \, \d \,  \frac{R_{\d}^{-1/2}}{|\nabla V (R_{\d})| }
\]
A straightforward computation shows that  for $\d=C \eps^{1/3}$, $\etabar$ is a supersolution of (\ref{1-1}) and we also have

\[
\sigma = \mathcal{O}( \eps^{1/3}) \quad \mbox{ and }\quad \g = \mathcal{O} \left( \ep^{1/6}  e^{ \ep^{-1/3}   \sqrt{R }} \right)  .
\]

Moreover, with this choice of $\d$, $\etabar^2> \lambda_\ep - V$   for every $|x| \leq R_{-\d }$, so using (\ref{etainf})
\[
\etaeps^2(x_0) \leq A = \lambda_\ep - V(x_0) < \etabar(x_0) .
\]

Because $\etaeps$ and $\etabar$ are going to zero at infinity, if the function $\etaeps-\etabar$ is positive somewhere  in $(r_0,\infty)$, for $r_0 := |x_0|$, then it attains a positive maximum at $ \tilde{r} \in (r_0,\infty)$, i.e. $\etaeps'(\tilde{r})=\etabar'(\tilde{r})$ and $\etaeps''(\tilde{r})<\etabar''(\tilde{r})$. Given the structure of (\ref{1-1}) and because $\etabar$ is a supersolution and $\etaeps$ a solution, if $V(\tilde{r})-\lambda_{\ep} \ge 0$ we would have that $\etaeps(\tilde{r}) \le \etabar(\tilde{r})$. In another hand, if $V(\tilde{r})-\lambda_{\ep} < 0$ then we would have $\etabar(\tilde{r})<\sqrt{\lambda_{\ep}-V(\tilde{r})}$, which for $\ep$ small enough, contradicts the definition of $\etabar$. Hence

\[
\etaeps(r) \leq \etabar(r) \quad\quad \mbox{ in } \hspace{1cm}  (r_0, \infty) \,.
\]

{\em Proof of (\ref{etaeps2}).} Using assumption (\ref{V2}), by exactly following \cite{AAB}, one finds that $|\etaeps - \sqrt{a_\ep^+}| \le C \ep^{1/3} \sqrt{a_\ep^+}$,
for $a_\ep := \lambda_\ep - V = a + \lambda_\ep - \lambda_0$. In view of (\ref{etaeps4}), this implies (\ref{etaeps2}).
 \\

{\em Proof of (\ref{etaeps3}).} For $x \in \R^2$ define $\etatil(y)=\etaeps(\ep(y-x))$ in $B_{2L}(x)$. This function satisfies
\[
\Delta \etatil = \etatil\, (V(\ep(y-x))+\etatil^2- \lambda_\ep) =: h_\ep
\]
After estimates  (\ref{etaeps1}) and (\ref{etaeps2}) $|h_\ep| \leq  C$, so using a H\"older estimate for the first derivative of $\etatil$ (see Theorem 8.32 in \cite{GT})  we have that $ \| \nabla \etatil \|_{L^{\infty} (B_{L}(x))} \leq  C $ for a constant $C$ independent of $x$ and hence the result. \\

{\em Step 4 : Proof of (\ref{etaeps6}).}\\

 We denote $L$  the elliptic operator obtained by linearizing equation (\ref{1-1})

 \[
    L := -\Delta + \frac 1{\ep^2} (V(x)+3\etaeps^2-\laeps),
 \]
and $\lambda_j$, $j=1,2,...$ , its eigenvalues in  $\R^2$. \\

Let $\mu$ be the first Dirichlet eigenvalue of $L$ in the half space $\Omega=\{x_1>0\}$ and $\psi$ the corresponding eigenfunction (which exists because of the compact embedding of $\mathcal{H}$ in $L^2$). Since $V$ and $\etaeps$ are radial, is clear that the odd extension of $\psi$ to $\R^2$ is a eigenfunction for $L$ in $\R^2$ with corresponding eigenvalue $\mu=\lambda_j$. Note that $j\ge 2$ because the odd extension change sign in $\R^2$.

We have that $L \etaeps=2\etaeps^4>0$ and $\etaeps>0$. Using the maximun principle due to Berestycki, Nirenberg and Varadhan  \cite{BNV}, this implies that the first eigenvalue of $L$ is positive. We will prove that if (\ref{etaeps6}) does not hold, then $\mu<0$, which contradicts the fact that $\lambda_1>0$.

Assume that $\etaeps'(r)>0$ at some $r\in (R_{-\delta_0}, R_{\delta_0})$. Then there exists
$\alpha < r < \beta$ such that $\etaeps'(\alpha) = \etaeps'(\beta) = 0$ and $\etaeps'>0$
in $(\alpha, \beta)$.
If $\alpha \le R_{-2\delta_0}$, then $\etaeps$ is increasing
on $(  R_{- 2\delta_0},   R_{-  \delta_0})$, so that $\etaeps( R_{- 2\delta_0})
\le  \etaeps(R_{- \delta_0})$. This is impossible for all sufficiently small $\ep$, since $\etaeps\to \sqrt{a^+}$
uniformly for $r < R_{-\ep^{1/3}}$, by (\ref{etaeps2}),  and $a^+(R_{-  2\delta_0}) > a^+(R_{- \delta_0})$.
Thus $\alpha \ge R_{-2\delta_0}$. The same argument, but using (\ref{etaeps1}) instead of (\ref{etaeps2}),
shows that $\beta \le R_{2\delta_0}$.

Now let $D := \{ x \in \R^2 : x_1>0, \alpha<|x|< \beta\}$.
Then
\[
\frac{\partial\etaeps}{\partial{x_1}}>0  \mbox{ in } D\, , \quad  \quad  \frac{\partial\etaeps}{\partial{x_1}}=0 \mbox{ in } \partial D \quad \mbox{ and } \quad L \left( \frac{\partial\etaeps}{\partial{x_1}}\right) = - \frac{\partial V}{\partial{x_1}} \etaeps \leq 0 \,\mbox{ in } D \,.
\]
The last inequality come from the differentiation of (\ref{1-1}) and hypothesis (\ref{V2}), which implies that
$\partial V/\partial R>0$ for $r\in (R_{-2\delta_0}, R_{2\delta_0})$. Using the monotonicity of Dirichlet eigenvalues with respect to the domain, this implies that $\mu<0$.

{\em Step 5 : Proof of (\ref{etaeps7}).}
For any $r \ge R$, define a function  $\etatil:(r,\infty)\to \R$ by
\[
\etatil(s) := \etaeps(r) \exp\left( \alpha  \frac{2}{p+2}(s^{\frac {p+2}2} - r^{\frac{p+2}2})\right)
\]
where $c_0$ and $p$ are the constants in (\ref{Vgrowth}).
It follows from  (\ref{etaeps4}) and (\ref{Vgrowth}) that if
$s \ge r$ and $r$ is sufficiently large, then $V(s) - \lambda_\ep + \etatil^2(s) \le V(s) \le c_0 s^p$,
so that if $r$ is sufficiently large, then
\[
-\Delta \etatil + \frac 1 {\ep^2}(V - \lambda_\ep + \etatil^2)\etatil \le
-\Delta \etatil(s) + \frac {c_0} {\ep^2}  s^p\etatil =
\left(( - \alpha^2+\frac{ c_0}{\ep^2}) s^p +\alpha( \frac p2 + 1)s^{\frac p2-1}
\right)\etatil.
\]
Choosing $\alpha = \frac{ (2c_0)^{1/2}}\ep$, it follows that $\etatil$ is a subsolution of (\ref{1-1})
in $(r,\infty)$
if $r$ is sufficiently large. For such $r$, noting that $\etatil(r) = \etaeps(r)$, we can
argue as in the proof of (\ref{etaeps1}) to deduce that $\etaeps- \etatil$ is nonnegative
in $(r,\infty)$.

Then since $\etatil(r) = \etaeps(r)$ and $\etatil(s)\le \etaeps(s)$ for $s \ge r$, we again use
(\ref{Vgrowth}) to conclude that
\[
\etaeps'(r)\  \ge \  \etatil'(r)  \ = \  -  \frac{ (2c_0)^{1/2}}\ep r^{\frac p2}\etaeps(r)  \ \ge  \ - \sqrt 2 \frac{c_0}\ep \sqrt {V(r)} \etaeps(r)
\]
for sufficiently large $r$.
On the other hand, by choosing $\alpha = \frac {\sqrt{c_0}}{2\ep}$ in the definition of $\etatil$, we obtain a decreasing supersolution (still denoted $\etatil$) such that $\etatil(r) = \etaeps(r)$. A similar application of the maximum
principle shows that $\etaeps$ is bounded above by (the new) $\etatil$ on $(r,\infty)$,
and in particular this implies that $\etaeps'(r)\le 0$. These facts combine to establish (\ref{etaeps7}).
\end{proof}

We next prove

\begin{lem} Assume that $V$ satisfies (\ref{V1}) and (\ref{V2}) and
the quadratic growth condition (\ref{Vquad}).
Let  $\etaeps$ be the positive minimizer found in Theorem \ref{theoeta}. Let $f_\ep (x) :=\xi_\ep (x)/ \etaeps^2 (x)$, where $\xi_\ep$ was defined in (\ref{eqxi}). Then there exists a constant $C$ independent of $\ep\in (0,\ep_1]$ such that

\begin{equation}
f_\ep(|x|) \le \left\{
\begin{array}{ll}
C \dist(x, \partial \cd) \ + C \ep^{2/3} \quad
&\mbox{ if }x\in \cd\\
C \ep^{2/3}&\mbox{ if }not.
\end{array}
\right.
\label{fep1}\end{equation}
In addition, for all sufficiently small $\ep$,
\be
\|\nabla \xi_\ep\|_\infty   \le C
\label{etaeps5}\ee
and
\be
\| f_\ep - f_0\|_\infty \le C \ep^{1/3}.
\label{f1}\ee
\label{lem:f}\end{lem}

\begin{proof}
For every $s\ge r\ge R_{\delta}$ (where $0<\d \le \d_0$ will be chosen later),
we define
\be\label{ubeta}
\etatil(s)=\etaeps(r)e^{-\mu_{\delta}(s^2-r^2)/2}
\quad\quad \hbox{ and }\quad\quad  \mu_{\delta}^2 = \frac{c_1(R_{\delta}^2-R^2) +(\lambda_\ep-\lambda_0)}{R_{\delta}^2 \ep^2}.
\ee
Using (\ref{Vquad}), where the constant $c_1$ is defined,  and arguing as in the proof of (\ref{etaeps7}),
we find that $\etatil-\etaeps$ is nonnegative in $(r,\infty)$.

We use the previous estimate and the definition of $\xi_\ep$ to compute
\[
f_\ep(r)=\frac
1{\etaeps^2(r)}\int_r^\infty s\etaeps^2 (s)\ ds\leq\int_r^\infty
e^{-\mu_{\delta}(s^2-r^2)} s\ ds=\frac 1 {2\mu_{\delta}} \quad\ \ \mbox{ for }r \ge R_\delta.
\]
The definition of $f_\ep$ implies that  $f_\ep'(r)=-r-2f_\ep(r)\frac {\etaeps '(r)}{\etaeps
(r)}$, and from the monotonicity (\ref{etaeps6}) of $\etaeps$, we infer that $f_\ep'(r)\geq-r$ in $(R_{-\d_0},R_{\d_0})$. Thus for any $R_{-\d_0}\le r \le R_{\d}$,
\[
f_\ep(r)\leq \frac{R_{\d}^2-r^2} 2 +\frac 1 {2\mu_{\delta}}.
\]
We now fix $\delta = \ep^{2/3}$, and we conclude from (\ref{V2}) and (\ref{etaeps4}) that (\ref{fep1}) holds
as long as $r \ge R_{-\delta_0}$.

For $0 \le r \le R_{-\d_0}$,  we write
\[
f_\ep(r) = \frac 1{\etaeps(r)^2} \int_r^{R_{-\d_0}} s \frac { \etaeps^2(s)}{\etaeps^2(r)} \, ds  +  \frac{\etaeps^2(R_{-\d_0})}{\etaeps^2(r)} f(R_{-\d_0})
\]
From  (\ref{etaeps2}) and (\ref{V2}), we see that if $0 \le r \le s \le   R_{-\ep^{1/3}}$, then
\be
\frac {\etaeps^2(s)}{\etaeps^2(r)} \le  \frac{(1+C\ep^{1/3})^2}{(1-C\ep^{1/3})^2} \frac {a^+(s)} {a^+(r)}
\le C \quad\mbox{ for sufficiently small }\ep,
\label{ratio}\ee
and by using the and the fact that $f_\ep(R_{-\delta_0}) \le C \ep^{2/3}+ C \delta_0$,
one easily deduces that (\ref{fep1}) holds for $r\in [0, R_{-\delta_0})$.\\

Next, the definition of $\xi_\ep$ implies that $|\nabla \xi_\ep(x)| = |x| \etaeps^2(x)$, so that (\ref{etaeps5}) follows from (\ref{etainf}) and (\ref{etaeps1}). \\

For  $r \ge  R_{-\ep^{1/3}}$, we see from (\ref{fep1}) that $|f_\ep(r) - f_0(r)| \le C \ep^{1/3} + |f_0(r)|$.
This is trivially bounded by $C\ep^{1/3}$ if $r\ge R$. If $R \le r \le R_{\ep^{-1/3}}$ then (\ref{V2}) implies that $c(R-r) \le a(r) \le C(R-r)$, and thus
\[
|f_0(r)| = f_0(r) \le \frac C{r-R}\int_r^Rs (R-s) ds \ \le C (R-r) \le C \ep^{1/3}.
\]

For $0\le r \le R_{-\ep^{1/3}} $ we write
\beann
    f_{\ep}(r) - f_0(r) &=& \left\{  \frac 1{\etaeps^2(r)} \int_r^{R_{-\ep^{1/3}}} s \etaeps^2(s) \, ds  -\frac 1{a(r)} \int_r^{R_{-\ep^{1/3}}} s a(s) \, ds  \right\}\\
    && + \frac {\etaeps^2(R_{-\ep^{1/3}})}{\etaeps^2(r)} f_\ep(R_{-\ep^{1/3}})  - \frac {a(R_{-\ep^{1/3}})}{a(r)} f_0(R_{-\ep^{1/3}}) \\
    &=& I+II  -III
\eeann
Using (\ref{ratio}) and our earlier estimates of $f_\ep, f_0$ for $r \ge R_{-\ep^{1/3}}$,
we see that
\[
    |II|\le C f_{\ep}(R_{-\ep^{1/3}})\le C \ep^{1/3} \hspace{.5cm} \mbox{and} \hspace{.5cm}| III| \le C  f_0(R_{-\ep^{1/3}}) \le C \ep^{1/3} \,.
\]
We further decompose the remaining term as
\[
I =  \left(\frac 1{\etaeps^2(r)}- \frac 1{a(r)}\right) \int_r^{R_{-\ep^{1/3}}} s \etaeps^2(s) \, ds  +\frac 1{a(r)} \int_r^{R_{-\ep^{1/3}}} s(\etaeps^2(s) -  a(s)) \, ds  .
\]
Using  (\ref{etaeps2}), it follows that
\[
|I| \le  C\ep^{1/3} \int_r^{R_{-\ep^{1/3}}} s\frac{ \etaeps^2(s)}{\etaeps^2(r)} \, ds  +
C \ep^{1/3} \int_r^{R_{-\ep^{1/3}}} s \frac{a(s)}{a(r)} \, ds  .
\]
Due to (\ref{etaeps6}), $\frac{ \etaeps^2(s)}{\etaeps^2(r)}\le 1$ if
$R_{-\delta_0}\le r \le s \le R_{-\ep^{1/3}}$. And if
$0\le r \le R_{-\delta_0}$ then $\etaeps^2(r)\ge C^{-1}$ and so
$\frac{ \etaeps^2(s)}{\etaeps^2(r)}\le C$. Thus the first integral is bounded by $C \ep^{1/3}$. The second integral is similarly estimated, using (\ref{V2}) in place
of (\ref{etaeps6}).
\end{proof}

\begin{rem}In the case of a potential $V$ for which (\ref{simplyconnected}) fails,
so that for example $\cd$ has the form $B_{R}\setminus B_{R'}$, one
expects that instead of being small, $f_\ep$ is large, namely,
$f_\ep \ge c e^{c/\ep}$ in the interior of $B_{R'}$. This is related
to the formation at very low rotations of a giant vortex in the
interior of $B_{R'}$. The arguments used to prove Lemma \ref{lem:f}
show in this situation that if $V$ grows quadratically in the
complement of $B_R$, as in (\ref{Vquad}), then $f_\ep$ is very small
in $\R^2\setminus B_R$. This suggests that at low rotations there
should be no vortices in $\R^2\setminus B_R$, but this cannot be
deduced from the arguments we use to prove Theorems \ref{T1} and
\ref{T2}.
\end{rem}
The last lemma in this section examines the case when $V$ has
subquadratic growth and $f_\ep$ is also large so that in principle
vortices could exist in the low density region.

\begin{lem} Assume that $V$ satisfies (\ref{V1}), (\ref{V2}) and \be
\mbox{ there exists $c_2>0$ and $p<2$ such that $V(r) \le c_2(r^p
+1)$ for all $r\ge R$.} \label{Vsubquad} \ee Then $f_{\ep}(x) \to
+\infty$ as $|x|\to \infty$. \label{lem:f}\end{lem}

Note that with these assumptions on $V$,  there is a sequence of
functions $\zeta_{\alpha}$ in $\mathcal{H}$ such that
$\inf_{\alpha}G_{\ep}(\zeta_{\alpha})= -\infty$.  Physically this
happens because the centrifugal force due to rotation is bigger than
the subquadratic trapping potential. This indicates that, although
one can prove that in this situation, $f_\ep\to \infty$ as $r\to
\infty $ (compare Lemma \ref{lem:f}), this is not expected to give
any information about the physical behaviour of condensates.

\begin{proof}
Let $q>2$. For every $r \ge \max\{1,R\}$, we claim that

\be\label{nu}
\etaeps(s)\geq \etaeps(r)e^{-\nu_{\ep,r}(s^q-r^q)/q} 
\ee

for all $s\ge r$. Where  $\nu_{\ep,r}$ is the positive root of the
polynomial $\nu^2-\frac{q}{r^q}\nu-\frac{c\, }{\ep^2r^{2q-2-p}}$,
which for $\ep$ small satisfy

\[
    \nu_{\ep,r} < C \, \ep^{-1} \, r^{-\beta}
\]

with $\beta=q-1-p/2$. Indeed, the right hand side of (\ref{nu}) is a subsolution in $(r,\infty)$ of (\ref{1-1}) while $\etaeps$ is a solution. Boths functions are going to zero at infinity  and they are equal at $s=r$, so the result come  arguing as in the proof of (\ref{etaeps1}). \\

We use the previous estimates and the definition of $\xi_\ep$ to
compute
$$
f_\ep(r)=\frac{\xi_\ep(r)}{\etaeps^2(r)}=\frac
1{\etaeps^2(r)}\int_r^\infty s\etaeps^2 (s)\ ds \geq \int_r^\infty
e^{-\nu_r(s^q-r^q)} s\ ds \ge\frac {r^{2-q}} {\nu_r} > C \, \ep \,
r^{1-p/2}.
$$
and hence the result.
\end{proof}


\section{ Splitting the energy}  \label{S:split}

In this section we recall the proofs of (\ref{split}) and
(\ref{Fep}).

For $U\subset \R^2$, we will write $E_\ep(w; U)$ etc to denote the
integrals over $U$ of the energy density appearing in the definition
of $E_\ep(u)=E_\ep(u;\R^2)$ , and similarly $G_\ep(\cdot, U),
F_\ep(\cdot, U)$.

Note that $v=u/\etaeps$ is well defined since $\etaeps>0$. Since
$\eta_\ep$ satisfies (\ref{1-1}), we multiply it by $\eta_\ep
(1-|v|^2)$ and integrate over a ball $B_r$ to find that
\[
\int_{B_r} (|v|^2-1) (-{1\over 4}\Delta \eta_\ep^2+{1\over
2\ep^2}\eta_\ep^2 (V(x)+ \eta_\ep^2)+\frac 12|\nabla \eta_\ep  |^2)
=  \frac {\lambda_\ep}{\ep^2}\int_{B_r}(|u|^2 - \etaeps^2)  .
\]
Note that the Lagrange multiplier term tends to $0$ as $r\to
\infty$, since  both the $L^2$ norms of $u$ and $\etaeps$ are 1.
Moreover,
\begin{eqnarray*}
E_\ep(v\eta_\ep;B_r)&=&
J_\eps(\etaeps;B_r)+F_\ep(v;B_r)+\int_{B_r}\frac12 |\nabla \eta_\ep
|^2(|v|^2-1)+ \frac 12 \etaeps\nabla \etaeps\cdot \nabla
|v|^2\\
& &\quad\quad -\frac 1{4\ep^2}\etaeps^4(1-|v|^2)^2+{1\over
4\eps^2}\eta^4|v|^4+ {1\over 2\eps^2}V(x)\eta^2|v|^2-{1\over
4\eps^2}\eta^4- {1\over 2\eps^2}V(x)\eta^2.
\end{eqnarray*}
We integrate by parts to obtain
\[
\int_{B_r} \frac 12 \etaeps\nabla \etaeps\cdot \nabla |v|^2 =
-\int_{B_r}\frac 14 |v|^2 \ \Delta \etaeps^2 + \int_{\partial
B_r}\frac 12 |v|^2  \etaeps \nu \cdot \nabla \eta
\]
We use (\ref{etaeps7}) to estimate
\[
|\int_{\partial B_r}\frac 12 |v|^2  \etaeps \nu \cdot \nabla \eta|
\le \frac C \ep \int_{\partial B_r}\frac 12 \eta^2  |v|^2 \sqrt V =
\frac C \ep  V(r)^{-1/2} \int_{\partial B_r} V\,  |u|^2.
\]
Since $\int_{\R^2} V |u|^2<\infty$, we can easily find a sequence
$r_k\to \infty$ such that the above integral tends to $0$. Combining
the above  and letting $r_k\to \infty$ along this sequence, we
obtain (\ref{split}).

The only property of $V$ that the above argument used (implicitly)
was   (\ref{Vgrowth}), which will be used in the proof of
(\ref{etaeps7}).

The integration by parts that leads to (\ref{Fep}) is justified in a
similar fashion. One must estimate boundary terms of the form
$\int_{\partial B_r} \xi  \nu \cdot (iv, \nabla v)$. To do this we
note that
\[
 \xi  \nu \cdot (iv, \nabla v)
\ = \ f_\ep(r) \etaeps^2 (iv, \nabla v) \ = \ f_\ep(r)  (iu, \nabla
u) \le \| f_\ep\|_\infty ( |u|^2+  |\nabla u|^2).
\]
We prove in (\ref{fep1}) that $f_\ep$ is bounded as long as $V$
satisfies (\ref{Vquad}) (in fact we show that $f_\ep \le C\ep^{2/3}$
for large $r$) and since $u\in H^1(\R^2)$, we can again find a
sequence $r_k\to \infty$ such that the boundary terms vanish.

Note also that the fact that $f_\ep \in L^\infty$, or equivalently
that $|\xi_\ep| \le C \etaeps^2$, implies that the term $ \xi_\ep
Jv$ appearing in (\ref{Fep}) is integrable on $\R^2$ for $v =
u/\etaeps$, whenever $u$ has finite energy.


\section{Proofs of Theorems \ref{T1} and \ref{T2}}

In this section we use the estimates  we have already established to complete
the proofs of our main theorems.


\begin{proof}[Proof of Theorem \ref{T1}]
We assume that $u_\ep$ minimizes $E_\ep$ and that  $\Omega \le C|\log\ep|$ is such that
(\ref{subcrit}) holds.

Let $\chi$ be a smooth function  such that $\chi \equiv 1$
in $ \{x\in \cd: \mbox{dist}(x,\partial \cd) \ge 2 |\log\ep|^{-3/2}\}$, and with support in $\cd_1$. We also
assume that $\|\nabla \chi\|_\infty \le 2 |\log\ep|^{3/2}$.

Let $v = u_\ep/\etaeps$, so that
$E_\ep(u) = G_\ep(\etaeps) + F_\ep(v)  = E_\ep(\etaeps) + F_\ep(v)$.
Thus $F_\ep(v) \le 0$. We write
\[
F_\ep(v) = A_1 - A_2+B
\]
where
\[
A_1 =
\int_{\R^2} \chi \left[
\frac{\etaeps^2}{ 2}|\nabla
v|^2   + \frac {\etaeps^4}{ 4\eps^2} (|v|^2-1)^2\ \right] dx,
\quad\quad\quad
A_2 =
2\O \int_{\R^2} \chi
 \xi_\ep\, Jv
   dx
\]
and
\[
B = \int_{\R^2} (1-\chi) \left[
 {\etaeps^2\over 2} \left( |\nabla
v|^2 - 4\O f_\ep \,Jv\right)
+ {\etaeps^4\over 4\eps^2} (|v|^2-1)^2\ \right] dx,
\]
It follows directly from our estimates on $f_\ep$ that $0< f_\ep  \le C (\ep^{2/3} +  |\log\ep|^{-3/2})$
in the support of $1-\chi$, for small enough $\ep$.
Since $\Omega \le  C|\log \ep|$, it follows that ${\O f_\ep} \le \frac 14$ for all sufficiently small $\ep$ and (recalling that $|Jv|\le \frac 12 |\nabla v|^2$) we deduce that
\[
 \left (|\nabla
v|^2- {4\O f_\ep} Jv \right )
\ge \frac 12 |\nabla v|^2
\]
in the support of $1-\chi$. It follows immediately that
\be
B \ge \int_{\R^2} (1-\chi) \left[
\frac{\etaeps^2}4 |\nabla v|^2
+ {\etaeps^4\over 4\eps^2} (|v|^2-1)^2\ \right] dx \ge 0
\label{B.est}\ee
and hence that $B=0$  if and only if $v$ is a constant of modulus $1$ in the support of $1-\chi$.

Since $F_\ep(v)\le 0$, it is clear that $A_1+B \le A_2$.

Next, define $\tilde \ep = {\ep } / ( \inf_{\cd_1} {\etaeps}) $, so that (in view of (\ref{etaeps2}) and the
definition of $\cd_1$)
\[
\tilde \ep \le  C\ep |\log \ep|^{3/4} , \quad\quad \frac1 {\tilde\ep^2} \le \frac{\etaeps^2}{\ep^2}\mbox{ in }\cd_1.
\]
Then  (\ref{B.est}) and  (\ref{etaeps2} imply that,
\be
\int_{\cd_1} \frac 12|\nabla v|^2 + \frac 1{4\tilde \ep^2}(|v|^2-1)^2  \le (\inf_{\cd_1} \etaeps)^{-2}(A_1+ 2B)
\le C |\log \, \ep|^{3/2} A_2.
\label{A2bounds}\ee

To continue, let $w = \frac v{|v|} = w^1+i w^2$. From (\ref{subcrit})  we see that $|v|\ge \frac 12$
in $\cd_1$, and hence it is clear that $w\in H^1(\cd_1)$, and $|w|^2\equiv 1$. It follows that $Jw= 0$; we will recall a standard proof of this fact in a moment.
Thus
\[
A_2 = 2\Omega \int_{\cd_1} \chi \xi_\ep (Jv - Jw) \ dx =
2\Omega \int_{\cd_1} \nabla^\perp(\chi \xi_\ep) \cdot [ (iv,\nabla v) -(iw, \nabla w)] \ dx.
\]
If we write $v = \rho e^{i\phi}$ in $\cd_1$, 
then a calculation shows that
\[
(iv,\nabla v) = \rho^2 \nabla \phi,\quad\quad (iw, \nabla w) = \nabla \phi.
\]
From the latter fact we see that $Jw = \frac 12 \nabla\times  (iw, \nabla w)  = 0$, as we asserted above.
Also, from this and the fact that $\rho \ge \frac 12$ in $\cd_1$ we estimate
\[
|(iv,\nabla v) - (iw, \nabla w)| = \frac{ |\rho^2 -1|}{\rho} |\rho \nabla \phi| \le 2 |\,|v|^2-1| \ |\nabla v|.
\]
Using  (\ref{A2bounds}) , we deduce that
\beann
A_2
&\le&
2\Omega \|\nabla (\chi \xi_\ep)\|_\infty \int_{\cd_1} \left( \frac  {\tilde \ep} 2|\nabla v|^2 + \frac 1{2\tilde \ep}(|v|^2-1)^2\right)dx
\\
&\le &
C \Omega \|\nabla (\chi \xi_\ep)\|_\infty  \ep  |\log \, \ep|^{9/4} A_2.
\eeann

One checks easily from the definitions and from
(\ref{etaeps5}) that
\be
\|\nabla (\chi \xi_\ep)\|_\infty \le \ \| \nabla \chi\|_\infty \|\xi_\ep\|_\infty +  \| \nabla \xi_\ep\|_\infty
\le C |\log \ep|^{3/2}
\label{chi-xi}\ee
so we conclude that $A_2\le  C\ep  |\log \ep|^{15/4}A_2 \le \frac 12 A_2$ for all sufficiently small $\ep$. We know from (\ref{A2bounds}) that $A_2\ge 0$, and it follows that $A_2=0$, and hence (again appealing to (\ref{A2bounds})) that $A_1 = B = 0$. Thus $\|\nabla v\|_{L^2} = \| 1 - |v|^2\|_{L^2} = 0$, and so $v$ is
a constant of modulus $1$ as required.

\end{proof}

The proof of Theorem \ref{T2} will use the following result, which is Lemma 8 in \cite{J}.

\begin{lem}  There exists a universal constant $C>0$ such
that for any $\kappa \in (1,2)$,  open set $U \subset \R^2$ and $u \in H^1(U; \R^2)$,
and $\ep\in ((0,1)$,
\begin{eqnarray}
&\left| \int_U \phi Ju \right|
& \le \kappa \int |\phi|  \frac{ e_\ep(u)}{|\log\ep|} \nonumber \\
& \quad\quad\quad
\quad &+C \ep^{(\kappa-1)/50}
(1 +\| \phi\|_{W^{1,\infty}})
\left( \|\phi\|_\infty + 1 + \int_{ \mbox{\scriptsize{supp}  }\,\phi } (|\phi| +1) e_\ep(u) \ dx\right)
\label{newJacest}\end{eqnarray}
for all   $\phi \in C^{0,1}_c(U)$. Here $e_\ep(u) = \frac 12 |\nabla u|^2 + \frac 1{4\ep^2} (|u|^2-1)^2$.
\label{L.newJacest}\end{lem}

The lemma as stated in \cite{J} does not explicitly specify the exponent
$(\kappa-1)/50$ appearing on the right-hand side of (\ref{newJacest}).
By inspection of the proof, however, one sees that this exponent can be
taken to have the form $\frac 12 \alpha$, where $\alpha = (\kappa-1)/12\kappa$ as in
Theorem 2.1 of \cite{js}.

\begin{proof}[Proof of Theorem \ref{T2}]
We continue to use notation from the proof of Theorem \ref{T1}, such as $A_1, A_2, B, \tilde \ep$, and so on.

We first  invoke the lemma, with
$\tilde \ep$ in place of $\ep$ and  $\chi\,\xi_\ep$ in place of $\phi$, and with $\kappa>1$ to be chosen.
This yields
 \[
\left| A_2 \right|
\  \le \
2\Omega\kappa\int_{\R^2} \chi\  \xi_\ep  \, \frac {e_{\tilde \ep}(v)}{|\log \tilde \ep|} dx + \ce,
 \]
where $\ce$ denotes the error terms in (\ref{newJacest}).
We note that for all sufficiently small $\ep>0$, the error term satisfies the bound
$\ce \le C \ep^\beta(1+ |A_2|)$, for $\beta = (\kappa-1)/100$, for all sufficiently small $\ep$. This is a consequence of (\ref{A2bounds}) and  the estimates
\[
\|\chi\ \xi_\ep\|_{W^{1,\infty}} \le C |\log\ep|^{3/2},\quad
\|\chi\ \xi_\ep\|_{L^{\infty}} \le C.
\]
These in turn follow from (\ref{chi-xi}) together with
(\ref{etaeps5}) .

Now the choice of $\tilde \ep$ implies that $e_{\tilde \ep}(v) \le
\frac  12|\nabla v|^2 + \frac {\etaeps^2}{4\ep^2}(|v|^2-1)^2$  in $\cd_1$,
and recalling that $\xi_\ep = f_\ep \etaeps^2$, we obtain
\begin{eqnarray}
(1- C\ep^\beta)|A_2|
\ &\le  \ &
2\Omega\kappa
\frac  {\|f_\ep\|_\infty} {|\log \tilde \ep|}
\int\chi (\frac {\etaeps^2}2|\nabla v|^2 + \frac {\etaeps^4}{4\ep^2}(|v|^2-1)^2)  + C \ep^\beta
\nnn\\
& = \ &
2\Omega\kappa
\frac  {\|f_\ep\|_\infty} {|\log \tilde \ep|}
 \ A_1 + C \ep^\beta.\nnn
\end{eqnarray}
We know from (\ref{f1}) that
$\|f_\ep\|_\infty \le (1+ C \ep^{1/3})\|f_0\|_\infty \le  (1+ C \ep^{\beta})\|f_0\|_\infty$, and from the choice of
$\tilde \ep$,  for any $K>0$ there exists $\ep_0>0$ such that
$|\log\tilde \ep| \ge (|\log\ep| - \log|\log\ep| )(1 + K \ep^\beta)$ if $0<\ep < \ep_0$. Thus
\[
|A_2| \le \Omega \left( \frac {2\| f\|_\infty}{|\log \ep| - \log|\log\ep|}\right) \kappa A_1 + C\ep^\beta
\]
for all sufficiently small $\ep$. Assume that $\Omega \le \frac 1 {2\| f\|_\infty}( |\log \ep| -(c_1+1) \log|\log\ep|)$, for $c_1$ to be chosen below. Then
\be
|A_2|
\ \le \ \left(1 -c_1\frac {\log|\log\ep|}{|\log \ep| -\log|\log\ep|}\right) \kappa A_1 + C\ep^\beta
\ \le \ \left(1 - c_1\frac {\log|\log\ep|}{|\log \ep| }\right) \kappa A_1 + C\ep^\beta.
\label{a2.kest}\ee
We now take $\kappa := 1+  c_1\frac {\log|\log\ep|}{|\log \ep| }$, so that $\beta = (\kappa-1)/100
= \frac {c_1} {100} \frac {\log|\log\ep|}{|\log \ep| }$. Recalling that $A_1+ B \le A_2$ and that $B\ge 0$, clearly $A_1\le A_2$, so
we deduce that
\[
c_1^2( \frac {\log|\log\ep|}{|\log \ep| })^2 A_1 \le C \ep^\beta = C |\log \ep|^{-c_1/100}.
\]
If $c_1 = 400$ then we conclude that $A_1 \le C |\log\ep|^{-2}$.

Then (\ref{a2.kest}) implies that $A_2\le C|\log\ep|^{-2}$, and it follows that
$B \le  C|\log\ep|^{-2}$.
In view of (\ref{A2bounds}), this implies that
\be
\int_{\cd_1}
 |\nabla v|^2
 +\frac 1{4\ep^2}(|v|^2-1)^2\ \le C  |\log\ep|^{-2}.
\label{cd1.est}\ee
The estimate $\| \nabla v\|_\infty \le \frac C \ep$ (see (\ref{etaeps3})) and (\ref{cd1.est}) are easily seen to imply that
\be
|v| \ge 1 - C |\log\ep|^{-1}\mbox{ in }\cd_1
\label{mod}\ee
for all sufficiently small $\ep$. Thus $\Omega$ is subcritical for small enough $\ep$.

\end{proof}


\end{document}